\documentclass[a4paper]{article}
\usepackage{a4wide}
\usepackage{amssymb,amsfonts,latexsym,stmaryrd}
\usepackage[PostScript=dvips]{diagrams}
\usepackage{epsfig}
\usepackage{tikz,ifdraft}
\usetikzlibrary{arrows,positioning,matrix,backgrounds,calc,fit}
\usepackage{QED}
\usepackage{eurosym}

\newtheorem{theorem}{Theorem}[section]

\newtheorem{proposition}[theorem]{Proposition}

\newcommand{\beqa}{\begin{eqnarray*}}
\newcommand{\eeqa}{\end{eqnarray*}\par\noindent}

\newcommand{\rarr}{\rightarrow}
\newcommand{\HH}{\mathcal{H}}

\newcommand{\ket}[1]{{|} #1\rangle}

\newcommand{\ie}{\textit{i.e.}~}

\newcommand{\Set}{\mathbf{Set}}

\newcommand{\Complex}{\mathbb{C}}

\newcommand{\Two}{\mathbf{2}}

\newcommand{\Real}{\mathbb{R}}

\newcommand{\id}{\mathsf{id}}

\newcommand{\Pow}{\mathcal{P}}

\newcommand{\IFF}{\; \Longleftrightarrow \;}
\newcommand{\AND}{\; \wedge \;}
\newcommand{\OR}{\; \vee \;}

\newcommand{\card}[1]{|#1|}

\newarrow{Eq}=====
\newarrow{dash}....>
\newarrow{mon}>--->
\newarrow{inc}C--->

\diagramstyle[newarrowhead=vee,newarrowtail=vee]

\newcommand{\op}{\mathsf{op}}

\newcommand{\Pos}{\mathbf{Pos}}

\newcommand{\Rpos}{\Real_{\geq 0}}

\newcommand{\supp}{\mathsf{supp}}
\newcommand{\DR}{\mathcal{D}_{R}}
\newcommand{\VR}{\mathcal{V}_{R}}

\newcommand{\rmap}[2]{\rho^{#1}_{#2}}

\newcommand{\mxi}{X^{(i)}}
\newcommand{\mxj}{X^{(j)}}
\newcommand{\mxk}{X^{(k)}}
\newcommand{\myj}{Y^{(j)}}
\newcommand{\myk}{Y^{(k)}}
\newcommand{\myi}{Y^{(i)}}

\newcommand{\Att}{\mathcal{A}}
\newcommand{\Atpos}{\mathbf{Att}}
\newcommand{\Tup}[1]{\prod_{a \in #1} D_a}
\renewcommand{\mid}{\, : \,}
\newcommand{\TS}{\mathcal{T}}
\newcommand{\RS}{\mathcal{R}}
\newcommand{\nj}{\bowtie}

\begin{document}

\title{Relational Databases and Bell's Theorem}
\author{Samson Abramsky}
\maketitle

\begin{abstract}
Our aim in this paper is to point out a surprising formal connection, between two topics which seem on face value to have nothing to do with each other:
relational database theory, and the study of non-locality and contextuality in the foundations of quantum mechanics.
We shall show that there is a remarkably direct correspondence between central results such as Bell's theorem in the foundations of quantum mechanics, and questions which arise naturally and have been well-studied in relational database theory.

\end{abstract}

\section{Introduction}

Our aim in this paper is to point out a surprising formal connection, between two topics which seem on face value to have nothing to do with each other:
\begin{itemize}
\item Relational database theory.
\item The study of non-locality and contextuality in the foundations of quantum mechanics.
\end{itemize}

We shall show, using the unified treatment of the latter developed in \cite{abramsky2011unified}, that there is a remarkably direct correspondence between central results such as Bell's theorem in the foundations of quantum mechanics, and questions which arise naturally and have been well-studied in relational database theory.

In particular, we shall see that the question of whether an ``empirical model'', of the kind which can be obtained by making observations of measurements performed on a physical system, admits a classical physical explanation in terms of a local hidden variable model, is mathematically equivalent to the question of whether a database instance admits a universal relation. The content of Bell's theorem and related results is that there are empirical models, predicted by quantum mechanics and confirmed by experiment, which do not admit such a universal relation. Moreover, while the original formulation of Bell's theorem involved probabilities, there are ``probability-free'' versions, notably Hardy's construction, which correspond directly to relational databases.

In fact, we shall show more broadly that there is a common mathematical language which can be used to described the key notions of both database theory, in the standard relational case and in a more general ``algebraic'' form covering e.g.~a notion of probabilistic databases, and also of the theory of non-locality and contextuality, two of the key quantum phenomena. These features are  central to current discussions of quantum foundations, and provide non-classical resources for quantum information processing.

The present paper is meant to be an introduction to these two topics, emphasizing their common content, presented in a manner which hopefully will be accessible to readers without prior knowledge of either.

How should this unexpected connection be interpreted?
One idea is that the notion of contextuality is rather fundamental, and we can see some outlines of a common `logic of contextuality' arising from this appearance of common structure in very different settings.

Ideally, some deeper connections can also be found, leading to interesting transfers of results and methods. 
A first step in this direction has already been taken, in joint work with Georg Gottlob and Phokion Kolaitis \cite{AGK}, in the closely related field of constraint satisfaction. An algorithmic question which arises naturally from the quantum side (see \cite{abramsky2010relational}) leads to a refined version of the constraint satisfaction paradigm, \emph{robust constraint satisfaction}, and to interesting new complexity results.

\section{Relational Databases}

\subsection{Review of basic notions}

We shall begin by reviewing some basic notions of relational database theory.

We start with an example to show the concrete scenario which is to be formalized.

\subsubsection*{Example}

Consider the following data table:
\begin{center}
\begin{tabular}{|l|l|l|l|} \hline
\textbf{branch-name} & \textbf{account-no} & \textbf{customer-name} & \textbf{balance} \\ \hline \hline
Cambridge & $10991$-$06284$ & Newton & \pounds 2,567.53 \\ \hline
Hanover & $10992$-$35671$ & Leibniz & \euro 11,245.75 \\ \hline
\ldots & \ldots & \ldots & \ldots \\ \hline
\end{tabular}
\end{center}

\noindent Let us anatomize this table. There are a set of \emph{attributes}, 
\[ \{ \mbox{\textbf{branch-name}}, \mbox{\textbf{account-no}}, \mbox{\textbf{customer-name}}, \mbox{\textbf{balance}} \} \]
which name the columns of the table. The entries in the table are `tuples' which specify a value for each of the attributes. The table is a set of such tuples.
A database will in general have a set of such tables, each with a given set of attributes.
The \emph{schema} of the database --- a static, syntactic  specification of the kind of information which can reside in the database --- is given by specifying the set of attributes for each of the tables.
The state of the database at a given time will be given by a set of tuples of the appropriate type for each of the tables in the schema.


We now proceed to formalize these notions.

We fix some set $\Att$ which will serve as a universe of attributes. A database schema $\Sigma$ over $\Att$ is a finite family $\Sigma = \{A_1, \ldots , A_k\}$ of finite subsets of $\Att$. 

At this --- surprisingly early! --- point, we come to an interesting juncture. There are two standard approaches to formalising the notion of relation  which can be found in the relational database literature. One --- the `unnamed perspective' \cite{ABV} ---  is to formalize the notion of tuple as an ordered $n$-tuple in $D^n$ for some set $D$ of data values; a relation is then a subset of $D^n$. This is motivated by the desire to make the connection to the standard notion of relational structure in first-order logic as direct as possible. This choice creates a certain distance between the formal notion of relation, and the informal notion of table; in practice this is not a problem.

For our purposes, however, we wish to make a different choice --- the `named perspective' \cite{ABV}: we shall formalize the notion of tuple, and hence of relation, in a fashion which directly reflects the informal notion. As we shall see, this will have both mathematical and conceptual advantages for our purposes. At the same time, there is no real problem in relating this formalism to the alternative  one  found in the literature. Note that the style of formalization we shall use is also commonly found in the older literature on relational databases, see e.g.~\cite{ullman1983principles}.

We shall assume that for each $a \in \Att$ there is a set $D_a$ of possible data values for that attribute.
Thus for example the possible values for \textbf{customer-name} should be character strings, perhaps with some lexical constraints; while for \textbf{balance} the values should be pairs (\textbf{currency}, \textbf{amount}), where \textbf{currency} comes from some fixed list (\pounds, \euro, \ldots), and \textbf{amount} is a number. These correspond to \emph{domain integrity constraints} in the usual database terminology.

Given $A \in \Sigma$, we define the set of $A$-tuples to be $\prod_{a \in A} D_a$.
Thus an $A$-tuple is a function which assigns a data value in $D_a$ to each $a \in A$.

In our example above, the first tuple in the table corresponds to the function
\[ \{ \textbf{branch-name} \mapsto \mbox{Cambridge},  \textbf{account-no} \mapsto 10991{-}06284, \textbf{customer-name} \mapsto  \mbox{Newton}, \]
\[ \textbf{balance} \mapsto \pounds 2,567.53 \} \]

A relation of type $A$ is a finite set of $A$-tuples. Given a schema $\Sigma$, an \emph{instance} of the schema, representing a possible state of the database, is given by specifying a relation of type $A$ for each $A \in \Sigma$.

\subsubsection*{Operations on relations}

We consider some of the fundamental operations on relations, which play a central r\^ole in relational databases. Firstly, relations of type $A$ live in the powerset $\Pow(\prod_{a \in A} D_a)$, which is a boolean algebra; so boolean operations such as union, intersection, and set difference can be applied to them.

Note that  the set of data values may in general be infinite, whereas the relations considered in database theory are finite.
Thus one must use set difference rather than an `absolute' notion of set complement.


Next, we consider the operation of projection. In the language of $A$-tuples, projection is \emph{function restriction}. That is, given an $A$-relation $R$, and a subset $B \subseteq A$, we define:
\[ R |_B := \{ t |_B \mid t \in R \} . \]
Here, since $t \in \Tup{A}$, $t |_B$ just means restriction of the function $t$ to $B$, which is a subset of its domain. This operation is then lifted pointwise to relations.

Now we consider the independent combination of relations, which is cartesian product in the standard formalism. The representation of tuples as functions leads to a `logarithmic shift' in the representation\footnote{Think of $2^{a}2^{b} = 2^{a+b}$, and hence $\log(xy) =  \log(x)+\log(y)$.}, whereby this operation is represented by \emph{disjoint union} of attribute sets.

Given an $A$-relation $R$ and a $B$-relation $S$, we form the disjoint union $A \sqcup B$, and the $A \sqcup B$-relation 
\[ R \otimes S \; := \; \{ t \in \Tup{A \sqcup B} \; : \; t|_A \in R \AND t|_B \in S \} . \]
Of course, as concrete sets $A$ and $B$ may overlap. We can force them to be disjoint by `tagging' them appropriately, e.g.
\[ A \sqcup B \; := \; \{ 0\} \times A \;\, \cup \;\, \{1\} \times B . \]
The minor housekeeping details of such tagging can safely be ignored.\footnote{The relevant result is the coherence theorem for monoidal categories \cite{mac1998categories}.}
We shall henceforth do so without further comment.

This is only a subset of the operations available in standard relational algebra \cite{ullman1983principles}. A more complete discussion could be given in the present setting, but this will suffice for our purposes.

\subsection{The functorial view}

We shall now show how the relational database formalism, in the style we have developed it, has a direct expression in functorial terms. This immediately brings a great deal of mathematical structure into play, and will allow us to relate some important database notions to concepts of much more general standing.

We shall assume the rudiments of the language of categories, functors and natural transformations.
All the background we shall need is covered in the charming (and succinct) text \cite{pierce1991basic}.

We shall consider the partial order $\Atpos$ of finite subsets of $\Att$, ordered by inclusion, as a category.

We shall define a functor $\TS : \Atpos^{\op} \rTo \Set$ where $\TS(A)$ is the set of $A$-tuples.
Formally, we define
\[ \TS(A) \; := \; \Tup{A} , \]
and if $A \subseteq B$, we define the \emph{restriction map} $\rmap{B}{A} : \TS(B) \rTo \TS(A)$ by
\[ \rmap {B}{A} : t \mapsto t|_A . \]
It is easy to verify functoriality of $\TS$, which means that, whenever $A \subseteq B \subseteq C$,
\[ \rmap {B}{A} \circ \rmap{C}{B} = \rmap{C}{A}, \]
and also that $\rmap{A}{A} = \id_{A}$.
Thus $\TS$ is a \emph{presheaf}, and restriction is exactly function restriction.

We also have the covariant powerset functor $\Pow : \Set \rTo \Set$, which acts on functions by direct image: if $f : X \rTo Y$, then
\[ \Pow f : \Pow X \rTo \Pow Y :: S \mapsto \{ f(x) \mid x \in S \} . \]
We can compose $\Pow$ with $\TS$ to obtain another presheaf 
\[ \RS := \Pow \circ \TS : \Atpos^{\op} \rTo \Set . \]
This presheaf assigns the set of $A$-relations to each set of attributes $A$; while the restriction map
\[ \rmap{B}{A} : \RS(B) \rTo \RS(A) \]
is exactly the operation of relation restriction, equivalent to the standard notion of projection in relation algebra, which we defined previously:
\[ \rmap{B}{A} : R \mapsto R|_A . \]

\subsubsection*{Natural Join}

One of the most important operations in relational algebra is \emph{natural join}. Given an $A$-relation $R$ and a $B$-relation $S$, we define an $(A \cup B)$-relation $R \nj S$:
\[ R \nj S \; := \; \{ t \in \Tup{A \cup B} \mid t|_A \in R \AND t|_B \in S \} . \]
We shall now show how this operation can be characterized in categorical terms.

Note firstly that since the powerset is naturally ordered by set inclusion, we can consider $\RS$ as a functor
\[ \RS : \Atpos^{\op} \rTo \Pos \]
where $\Pos$ is the category of posets and monotone maps. $\Pos$ is order-enriched; given monotone maps $f, g : P \rarr Q$, we can define the pointwise order:
\[ f \leq g \; \equiv \; \forall x \in P. \, f(x) \leq g(x) . \]
Now suppose we are given attribute sets $A$ and $B$. We consider the following diagram arising from the universal property of product in $\Set$.
\begin{diagram}[8em]
\RS(A) & \lTo^{\pi_1} & \RS(A) \times \RS(B) & \rTo^{\pi_2} & \RS(B) \\
& \luTo_{\rmap{A \cup B}{A}} & \udash^{\langle \rmap{A \cup B}{A}, \rmap{A \cup B}{B} \rangle} & \ruTo_{\rmap{A \cup B}{B}} \\
& & \RS(A \cup B) &&
\end{diagram}

\begin{proposition}
The natural join $\nj : \RS(A) \times \RS(B) \rarr \RS(A \cup B)$ is uniquely characterized as the left adjoint of $\langle \rmap{A \cup B}{A}, \rmap{A \cup B}{B} \rangle$; that is, as the unique map satisfying
\[ \id_{\RS(A \cup B)} \; \leq \; \; \nj \circ \, \langle \rmap{A \cup B}{A}, \rmap{A \cup B}{B} \rangle, \qquad
\langle \rmap{A \cup B}{A}, \rmap{A \cup B}{B} \rangle \, \circ \nj \;\; \leq \; \id_{\RS(A) \times \RS(B)} . \]
\end{proposition}

The fact that in general a relation $R \in \RS(A \cup B)$ satisfies only
\[ R \; \subseteq \; R|_A \nj R|_B , \]
with strict inclusion possible, corresponds to the fact that natural join is in general a `lossy' operation. Lossless joins correspond exactly to the case when equality holds.

\subsection{The sheaf-theoretic view}
\label{sheafsec}

We shall now show, building on the presheaf structure described in the previous sub-section, how a number of important database notions can be interpreted geometrically, in the language of sheaves and presheaves.

\subsubsection*{Schemas as covers and gluing conditions}
We shall interpret a schema $\Sigma = \{A_1, \ldots , A_k\}$ of finite subsets of $\Att$ as a \emph{cover}.
That is, we think of the attribute sets $A_i$ as `open sets' expressing some local information in the sense of related clusters of attributes; these sets cover $A := \bigcup_{i=1}^k A_i$, the global set of attributes for the schema. Conversely, we think of the global set $A$ as being decomposed into the local clusters $A_i$; which is exactly the standard point of view in databases.

The basic idea of sheaf theory is to analyze the passage from local to global behaviour in mathematical structures. A number of important notions in databases have exactly this character, and can be described naturally in sheaf-theoretic terms.

An instance $(R_1, \ldots , R_k)$ of a schema $\Sigma$ is given by specifying a relation $R_i \in \RS(A_i)$ for each $A_i \in \Sigma$. In sheaf-theoretic language, this is a family of local sections, defined over the open sets in the cover. A central issue in geometric terms is whether we can glue these local sections together into a global section defined over $A := \bigcup_{i=1}^k A_i$.

More precisely, we can ask:
\begin{center}
Does there exist a relation $R \in \RS(A)$ such that $R |_{A_i} = R_i$, $i = 1, \ldots , k$. 
\end{center}
We say that the \emph{gluing condition} is satisfied for the instance $(R_1, \ldots , R_k)$ if such a relation exists.

This has been studied as an algorithmic question in database theory, where it is referred to as the \emph{join consistency property}; it is shown in \cite{honeyman1980testing} that it is NP-complete.

Note that a \emph{necessary condition} for this to hold is that, for all $i, j = 1, \ldots , k$:
\begin{equation}
\label{compeq}
R_i |_{A_i \cap A_j} \; = R_j |_{A_i \cap A_j} . 
\end{equation}
Indeed, if such an $R$ exists, then 
\[ R_i |_{A_i \cap A_j} = (R |_{A_i}) |_{A_i \cap A_j} = R |_{A_i \cap A_j} , \]
using the functoriality of restriction, and similarly for $R_j |_{A_i \cap A_j}$.

We shall say that a database instance $(R_1, \ldots , R_k)$ for which this condition~(\ref{compeq}) holds has consistent projections, and refer to the family of relations in the instance as a \emph{compatible family}.

These notions can be generalized to apply to any presheaf. If the gluing condition can \emph{always} be satisfied, for any cover and any family of compatible elements, and moreover there is a \emph{unique} element which satisfies it, then the presheaf is a sheaf.

It is of course a well-known fact of life in databases, albeit expressed in a different language, that our relational presheaf $\RS$ is \emph{not} a sheaf.

In fact, we have the following:
\begin{proposition}
An instance $(R_1, \ldots , R_k)$ satisfies the gluing condition if and only if there is a universal relation $R$ for the instance.
\end{proposition}

Here we take a universal relation for the instance by definition to be a relation defined on the whole set of attributes from which each of the relations in the instance can be recovered by projection.
This notion, and various related ideas,  played an important r\^ole in early developments in relational database theory; see e.g. \cite{maier1984foundations,fagin1982simplied,korth1984system,maier1983maximal,ullman1983principles}.

Thus the standard notion of universal relation in databases corresponds exactly to the standard notion of solution to the gluing condition in sheaf theory, for the particular case of the relational presheaf $\RS$.

It is also standard that a universal relation need not exist in general, and even if it exists, it need not be unique.
There is a substantial literature devoted to the issue of finding conditions under which these properties do hold.

There is a simple connection between universal relations and lossless joins.

\begin{proposition}
Let $(R_1, \ldots , R_k)$ be an instance for the schema $\Sigma = \{ A_1, \ldots , A_k \}$. Define $R := \; \nj_{i=1}^k R_i$.
Then a universal relation for the instance exists if and only if $R |_{A_i} = R_i$, $i = 1, \ldots , k$, and in this case $R$ is the largest relation in $\RS(\bigcup_i A_i)$ satisfying the gluing condition.
\end{proposition}
\begin{proof}
We note that, if a relation $S$ satisfies $S|_{A_i} = R_i$, $i = 1, \ldots , k$, then $S \subseteq \; \nj_{i=1}^k R_i$ by the adjoint property of the natural  join. Moreover, since projection is monotone, in this case $R_i \subseteq S|_{A_i} \subseteq (\nj_{i=1}^k R_i) |_{A_i} \subseteq R_i$.
\end{proof}

There are further categorical aspects of relational databases which it might prove interesting to pursue. In particular, one can define categories of schemas and of instances and their morphisms, and the construction of colimits in these categories may be applicable to issues of data integration.
However, we shall not pursue these ideas here.
Instead, we will turn to a natural generalization of relational databases which arises rather effortlessly from the formalism we have developed to this point.

\section{Algebraic Databases}

We begin by revisiting the definition of the relational presheaf $\RS$ in terms of the covariant powerset functor $\Pow$. An alternative presentation of subsets is in terms of \emph{characteristic functions}.
That is, we have the familiar isomorphism $\Pow(X) \cong \Two^X$, where $\Two := \{ 0, 1 \}$ is the 2-element boolean algebra.

We can also use this representation to define the functorial action of powerset. Given $s : X \rarr \Two$ and $f : X \rarr Y$, we define $f^*(s) : Y \rarr \Two$ by
\begin{equation}
\label{matimdef}
f^*(s) : y \mapsto \bigvee_{f(x)=y} s(x) . 
\end{equation}
It is easy to see that this is equivalent to
\[ f^*(s)(y) = 1 \IFF \exists x \in S. \, f(x) = y . \]
Here $S$ is the subset of $X$ whose characteristic function is $s$.

We can specialise this to the case of an inclusion function $\iota : A \rinc B$ which induces a map $\Two^B \rarr \Two^A$ by restriction:
\[ s : B \rarr \Two \;\; \mapsto \;\; (s |_A) : A \rarr \Two . \]
What we obtain in this case is exactly the notion of \emph{projection} of a relation, as defined in the previous section. 

The advantage of this `matrix' style of definition of the powerset is that it can immediately be generalized rather widely. There is a minor caveat. In the above definition, we used the fact that $\Two$ is a \emph{complete} boolean algebra, since there was no restriction on the cardinality of the preimages of $f$. In the database context, of course, all sets are typically finite.\footnote{The sets of data values $D_a$ may be infinite, but only finitely many values will appear in a database instance.} We shall enforce a finiteness condition explicitly in our general definition.

We recall that a \emph{commutative semiring} is a structure $(R, +, 0, \cdot, 1)$, where $(R, +, 0)$ and $(R, \cdot, 1)$ are commutative monoids, and moreover multiplication distributes over addition:
\[ x \cdot (y + z) = x \cdot y + x \cdot z . \]
Many examples of commutative semirings arise naturally in Computer Science: we list a few of the most common.
\begin{itemize}
\item The reals
\[ (\Real, +, 0, {\times}, 1). \]
More generally, any commutative ring is a commutative semiring.
\item  The non-negative reals
\[ (\Rpos, +, 0, {\times}, 1). \]
\item The booleans
\[ \Two = (\{ 0, 1 \}, \vee, 0, \wedge, 1) . \]
More generally, \emph{idempotent} commutative semirings are exactly the distributive lattices.
\item The min-plus semiring
\[ (\Rpos \cup \{\infty\}, {\min}, \infty, +, 0). \]
\end{itemize}
We also note the r\^ole played by \emph{provenance semirings} in database theory \cite{green2007provenance,buneman2007provenance,DBLP:journals/ftdb/CheneyCT09}.

We fix a semiring $R$. Given a set $X$, the \emph{support} of a function $v : X \rarr R$ is the set of $x \in X$ such that $v(x) \neq 0$. We write $\supp(v)$ for the support of $v$. We shall write $\VR(X)$ for the set of functions $v : X \rarr R$ of \emph{finite} support. We shall write $\DR(X)$ for the subset of $\VR(X)$ of those functions $d : X \rarr R$ such that
\[ \sum_{x \in X} d(x) = 1 . \]
Note that the finite support condition ensures that this sum is well-defined.

We shall refer to elements of $\VR(X)$ as $R$-valuations on $X$, and of $\DR(X)$ as $R$-distributions.

We consider a few examples:
\begin{itemize}
\item If we take $R = \Two$, then $\VR(X)$ is the set of finite subsets of $X$, and $\DR(X)$ is the set of finite non-empty subsets.
\item If we take $R = (\Rpos, +, 0, {\times}, 1)$, then $\DR(X)$ is the set of discrete (finite-support) probability distributions on $X$.
\end{itemize}
Algebraically, $\VR(X)$ is the free $R$-semimodule over the set $X$ \cite{golan1999semirings}.

These constructions extend to functors on $\Set$. Given $f : X \rarr Y$, we define
\[ \VR(f) : \VR(X) \rarr \VR(Y) :: v \mapsto [ y \mapsto \sum_{f(x) =y} v(x) ] . \]
This restricts to $\DR$ in a well-defined fashion.
Taking $R = \Two$, we see that $\VR(f)$ is exactly the direct image of $f$, defined as in~(\ref{matimdef}).


We can now generalize databases from the standard relational case to `relations valued in a semiring' by replacing $\Pow$ by $\VR$ (or $\DR$) in our definition of $\RS$; that is, we take $\RS := F \circ \TS$, where $F$ is $\VR$ or $\DR$ for some commutative semiring $R$. We recover the standard notion exactly when $R = \Two$. In the case where $R = (\Rpos, +, 0, {\times}, 1)$ and $F = \DR$, we obtain a notion of probabilistic database, where each relation specifies a probability distribution over the set of tuples for its attribute-set.

Moreover, our descriptions of the key database operations all generalise to any semiring.
If we apply the definition of the functorial action of $\VR$ or $\DR$ to the case of restriction maps induced by inclusions, we obtain the right notion of \emph{generalised projection}, which can be applied to any algebraic database.
We have already seen that we recover the standard notion of projection in the Boolean case.
In the case where the semiring is the non-negative reals, so we are dealing with probability distributions, projection is exactly \emph{marginalization}.

We also note an important connection between probabilistic and relational databases. We can always pass from a probabilistic to a relational instance by taking the \emph{support} of the distribution.
Algebraically, this corresponds to mapping all positive probabilities to $1$; this is in fact the action of the unique semiring homomorphism from the non-negative reals to the booleans.

In general, many natural properties of databases will be \emph{preserved} by this homomorphic mapping.
This means that if we show that such a property is \emph{not} satisfied by the support, we can conclude that it is not satisfied by the probabilistic instance. Thus we can leverage negative results at the relational level, and lift them to the probabilistic setting.

We shall see a significant example of a probabilistic database in the next section.

\section{From databases to observational scenarios}

We shall now offer an alternative interpretation of the relational database formalism, with a very different motivation. This will expose a surprising connection between database theory, and on face value a completely different topic, namely Bell's theorem in the foundations of quantum mechanics \cite{bell1964einstein}.

Our starting point is the idealized situation depicted in the following diagram.

\begin{center}
\begin{tikzpicture}[scale=2.54]
\ifx\dpiclw\undefined\newdimen\dpiclw\fi
\global\def\dpicdraw{\draw[line width=\dpiclw]}
\global\def\dpicstop{;}
\dpiclw=0.8bp
\dpiclw=2bp
\dpicdraw[fill=lightgray](0,-0.301272) rectangle (0.903816,0.301272)\dpicstop
\dpiclw=1bp
\dpicdraw (0.451908,-0.524213) circle (0.047444in)\dpicstop
\draw (0.451908,-0.403705) node[above=-1.807632bp]{$a$};
\draw (0.572417,-0.524213) node[right=-1.807632bp]{$b$};
\draw (0.451908,-0.644722) node[below=-1.807632bp]{$c$};
\draw (0.331399,-0.524213) node[left=-1.807632bp]{$d$};
\draw (0.451908,-0.524213) node{$\mathbf{\cdot}$};
\dpicdraw[fill=black](0.397679,-0.578442) circle (0.009489in)\dpicstop
\dpiclw=2bp
\dpicdraw[fill=lightgray](2.108904,-0.301272) rectangle (3.01272,0.301272)\dpicstop
\dpiclw=1bp
\dpicdraw (2.560812,-0.524213) circle (0.047444in)\dpicstop
\draw (2.560812,-0.403705) node[above=-1.807632bp]{$a'$};
\draw (2.681321,-0.524213) node[right=-1.807632bp]{$b'$};
\draw (2.560812,-0.644722) node[below=-1.807632bp]{$c'$};
\draw (2.440304,-0.524213) node[left=-1.807632bp]{$d'$};
\draw (2.560812,-0.524213) node{$\mathbf{\cdot}$};
\dpicdraw[fill=black](2.615041,-0.469984) circle (0.009489in)\dpicstop
\dpicdraw (0.150636,0.331399)
 ..controls (0.150636,0.497787) and (0.28552,0.632671)
 ..(0.451908,0.632671)
 ..controls (0.618296,0.632671) and (0.75318,0.497787)
 ..(0.75318,0.331399)\dpicstop
\dpicdraw (0.692926,0.331399)
 --(0.75318,0.331399)\dpicstop
\dpicdraw (0.687659,0.38151)
 --(0.746597,0.394037)\dpicstop
\dpicdraw (0.672089,0.42943)
 --(0.727134,0.453938)\dpicstop
\dpicdraw (0.646895,0.473066)
 --(0.695642,0.508483)\dpicstop
\dpicdraw (0.61318,0.51051)
 --(0.653498,0.555288)\dpicstop
\dpicdraw (0.572417,0.540127)
 --(0.602544,0.592308)\dpicstop
\dpicdraw (0.526387,0.560621)
 --(0.545006,0.617926)\dpicstop
\dpicdraw (0.477101,0.571097)
 --(0.4834,0.631021)\dpicstop
\dpicdraw (0.426715,0.571097)
 --(0.420417,0.631021)\dpicstop
\dpicdraw (0.37743,0.560621)
 --(0.35881,0.617926)\dpicstop
\dpicdraw (0.331399,0.540127)
 --(0.301272,0.592308)\dpicstop
\dpicdraw (0.290636,0.51051)
 --(0.250318,0.555288)\dpicstop
\dpicdraw (0.256921,0.473066)
 --(0.208174,0.508483)\dpicstop
\dpicdraw (0.231727,0.42943)
 --(0.176682,0.453938)\dpicstop
\dpicdraw (0.216157,0.38151)
 --(0.15722,0.394037)\dpicstop
\dpicdraw (0.21089,0.331399)
 --(0.150636,0.331399)\dpicstop
\dpicdraw (0.451908,0.331399)
 --(0.598149,0.437649)\dpicstop
\filldraw (0.615857,0.413276)
 --(0.695642,0.508483)
 --(0.58044,0.462023) --cycle
\dpicstop
\dpicdraw (2.25954,0.331399)
 ..controls (2.25954,0.497787) and (2.394424,0.632671)
 ..(2.560812,0.632671)
 ..controls (2.7272,0.632671) and (2.862084,0.497787)
 ..(2.862084,0.331399)\dpicstop
\dpicdraw (2.80183,0.331399)
 --(2.862084,0.331399)\dpicstop
\dpicdraw (2.796563,0.38151)
 --(2.855501,0.394037)\dpicstop
\dpicdraw (2.780993,0.42943)
 --(2.836038,0.453938)\dpicstop
\dpicdraw (2.7558,0.473066)
 --(2.804547,0.508483)\dpicstop
\dpicdraw (2.722085,0.51051)
 --(2.762403,0.555288)\dpicstop
\dpicdraw (2.681321,0.540127)
 --(2.711448,0.592308)\dpicstop
\dpicdraw (2.635291,0.560621)
 --(2.65391,0.617926)\dpicstop
\dpicdraw (2.586006,0.571097)
 --(2.592304,0.631021)\dpicstop
\dpicdraw (2.535619,0.571097)
 --(2.529321,0.631021)\dpicstop
\dpicdraw (2.486334,0.560621)
 --(2.467714,0.617926)\dpicstop
\dpicdraw (2.440304,0.540127)
 --(2.410176,0.592308)\dpicstop
\dpicdraw (2.39954,0.51051)
 --(2.359222,0.555288)\dpicstop
\dpicdraw (2.365825,0.473066)
 --(2.317078,0.508483)\dpicstop
\dpicdraw (2.340632,0.42943)
 --(2.285587,0.453938)\dpicstop
\dpicdraw (2.325062,0.38151)
 --(2.266124,0.394037)\dpicstop
\dpicdraw (2.319795,0.331399)
 --(2.25954,0.331399)\dpicstop
\dpicdraw (2.560812,0.331399)
 --(2.439858,0.465732)\dpicstop
\filldraw (2.462247,0.485892)
 --(2.359222,0.555288)
 --(2.417469,0.445573) --cycle
\dpicstop
\draw (0.451908,-0.825485) node[below=-1.807632bp]{Alice};
\draw (2.560812,-0.825485) node[below=-1.807632bp]{Bob};
\end{tikzpicture}

\end{center}

There are several agents or experimenters, who can each select one of several different measurements $a, b, c, d, \ldots$ to perform, and  observe one of several different outcomes. These agents may or may not be spatially separated. When a system is prepared in a certain fashion and measurements are selected, some corresponding outcomes will be observed. These individual occurrences or `runs' of the system are the basic events.
Repeated runs allow relative frequencies to be tabulated, which can be summarized by a probability distribution on events for each selection of measurements. We shall call such a family of probability distributions, one for each choice of measurements,  an \emph{empirical model}.

As an example of such a model, consider the following table.
\begin{center}
\begin{tabular}{ll|ccccc}
 &  & $(0, 0)$ & $(1, 0)$ & $(0, 1)$ & $(1, 1)$  &  \\ \hline
$a$ & $b$ & $1/2$ & $0$ & $0$ & $1/2$ & \\
$a'$ & $b$ & $3/8$ & $1/8$ & $1/8$ & $3/8$ & \\
$a$ & $b'$ & $3/8$ & $1/8$ & $1/8$ & $3/8$ &  \\
$a'$ & $b'$ & $1/8$ & $3/8$ & $3/8$ & $1/8$ & 
\end{tabular}
\end{center}
The intended scenario here is that Alice can choose between measurement settings $a$ and $a'$, and Bob can choose $b$ or $b'$.
These will correspond to different quantities which can be measured.\footnote{For example, in the quantum case these settings may correspond to different directions along which to measure `Spin Up' or `Spin Down' \cite{yanofsky2008quantum}.}
We assume that these choices are made independently. Thus the \emph{measurement contexts}  are
\[ \{ a, b \}, \quad \{ a', b\}, \quad \{ a, b' \}, \quad \{ a', b' \} ,\]
and these index the rows of the table.
Each measurement has possible outcomes $0$ or $1$.

Note that, with a small change of perspective, we can see this in database terms.
Take the global set of attributes $\Att = \{ a, a', b, b' \}$, and consider the schema
\[ \Sigma := (\{ a, b \}, \; \{ a', b\}, \; \{ a, b' \}, \; \{ a', b' \}) . \]
For each $a \in \Att$, we take $D_a := \{ 0, 1\}$.

For each $A \in \Sigma$, we have a `table' in the algebraically generalized sense discussed in the previous section. That is, we have a distribution $d_A \in \DR \circ \TS(A)$, where $R = \Rpos$ is the semiring of non-negative reals. Thus $d_A$ is a probability distribution on $\TS(A)$, the set of $A$-tuples.

To make a direct connection with standard relational databases, we can pass to the support of the above table, which yields the following:
\begin{center}
\begin{tabular}{ll|ccccc}
 &  & $(0, 0)$ & $(1, 0)$ & $(0, 1)$ & $(1, 1)$  &  \\ \hline
$a$ & $b$ & $1$ & $0$ & $0$ & $1$ & \\
$a'$ & $b$ & $1$ & $1$ & $1$ & $1$ & \\
$a$ & $b'$ & $1$ & $1$ & $1$ & $1$ &  \\
$a'$ & $b'$ & $1$ & $1$ & $1$ & $1$ & 
\end{tabular}
\end{center}
This corresponds to the instance of the schema $\Sigma$ where for each $A = \{ \alpha, \beta \} \in \Sigma \setminus \{\{ a, b \}\}$,
there is the `full' table of all possible tuples:
\begin{center}
\begin{tabular}{| c | c |} \hline
$\alpha$ & $\beta$ \\ \hline \hline
 0 & 0 \\ \hline
 0 & 1 \\ \hline
 1 & 0 \\ \hline
 1 & 1 \\ \hline
\end{tabular}
\end{center}
while for $\{ a, b \}$ we have the table with only two tuples:
\begin{center}
\begin{tabular}{| c | c |} \hline
$a$ & $b$ \\ \hline \hline
 0 & 0 \\ \hline
 1 & 1 \\ \hline
\end{tabular}
\end{center}

Thus we have a formal passage between empirical models and relational databases.
To go further, we must understand how empirical models such as these can be used to draw striking conclusions about the foundations of physics.

\section{Empirical models and hidden variables}

Most of our discussion is independent of any particular physical theory.
However, it is important to understand how quantum mechanics, as our most highly confirmed theory,  gives rise to a class of empirical models of the kind we have been discussing.

To obtain such a model, we must provide the following ingredients:
\begin{itemize}
\item A quantum state.
\item For each of the `measurement settings', which correspond to attributes in database terms, a physical observable or measurable quantity. Each such observable will have a set of associated possible outcomes, which will correspond to the set of data values associated with that attribute.
\end{itemize}
The `statistical algorithm' of quantum mechanics will then prescribe a probability for each measurement outcome when the given state is measured with that observable.

Although we shall not really need the details of this, we briefly recall some basic definitions.
For further details, see e.g.~\cite{nielsen2000quantum,yanofsky2008quantum}.

\subsection*{A crash course in qubits}

Whereas a classical bit register has possible states $0$ or $1$, a qubit state is given by a \emph{superposition} of these states. More precisely, a (pure) qubit state is given by a vector in the 2-dimensional complex vector space $\Complex^2$, \ie a complex linear combination $\alpha_0 \ket{0} + \alpha_1 \ket{1}$, subject to the normalization constraint $|\alpha_0|^2 +  |\alpha_1|^2 = 1$. Here $\ket{0}, \ket{1}$ is standard Dirac notation for the basis vectors $[1, 0]^T$ and $[0, 1]^T$.

Measurement of such a state (in the $\ket{0}$, $\ket{1}$ basis) is inherently probabilistic; we get $\ket{i}$ with probability $|\alpha_i|^2$. 

There is a beautiful geometric picture of this complex 2-dimensional geometry in real three-dimensional space. This is the Bloch sphere representation:

\begin{center}
\epsfig{figure=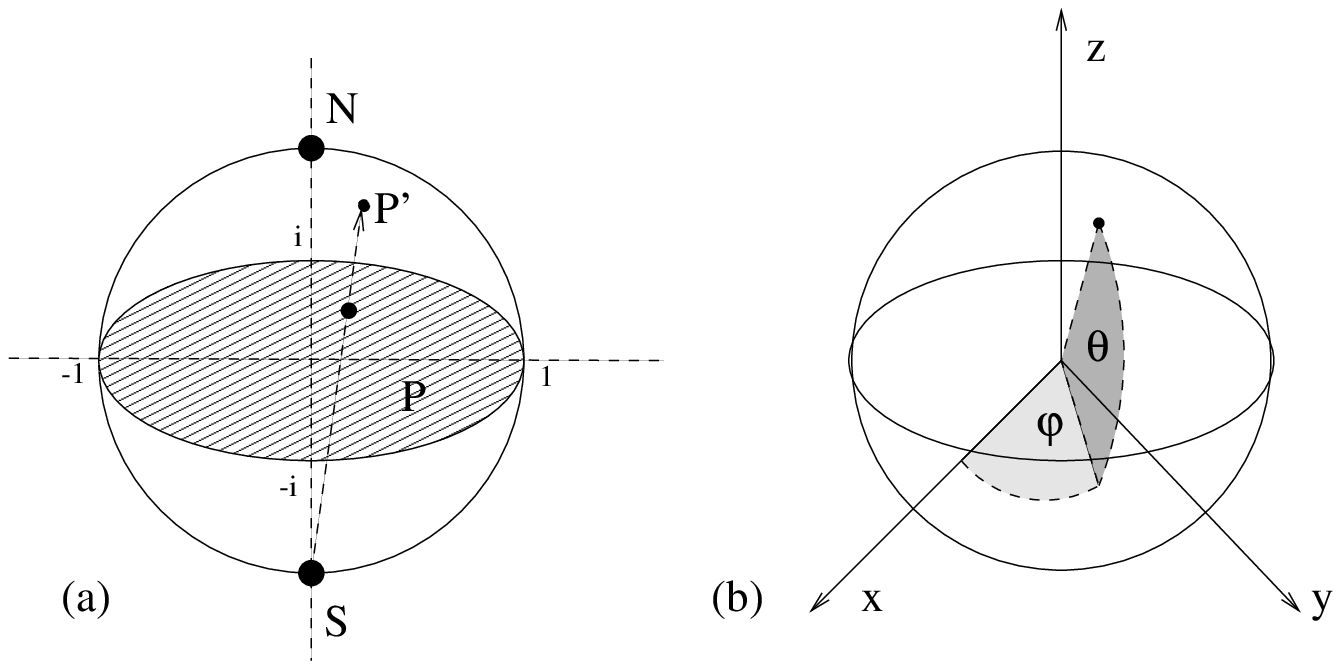,width=2.5in}
\end{center}

The pure qubit states correspond to points on the surface of the sphere.
However, this one-qubit case does not yet provide non-classical resources for information processing.
Things get interesting with $n$-qubit registers
\[ \sum_i \alpha_i \ket{i}, \qquad i \in \{ 0, 1 \}^n . \]
It is at this point, in particular, that \emph{entanglement phenomena} arise.

A typical example of an entangled state is the Bell state:
\begin{center}
\begin{tikzpicture}[scale=2.54]
\ifx\dpiclw\undefined\newdimen\dpiclw\fi
\global\def\dpicdraw{\draw[line width=\dpiclw]}
\global\def\dpicstop{;}
\dpiclw=0.8bp
\dpicdraw[fill=white!10!black](0.125,0) circle (0.049213in)\dpicstop
\dpicdraw[fill=white!10!black](1.875,0) circle (0.049213in)\dpicstop
\dpicdraw (0.25,0)
 --(1.75,0)\dpicstop
\draw (1,0) node[above=-0.75bp]{$|00 \rangle + |11\rangle$};
\end{tikzpicture}

\end{center}

We can think of two particles, each with a qubit state, held by Alice and Bob. However, these two particles are entangled. If Alice measures her qubit, then if she gets the answer $\ket{0}$, the state will collapse to $\ket{00}$, and if Bob measures his qubit, he will get the answer $\ket{0}$ with certainty; similarly if the result of Alice's measurement is $\ket{1}$.
This non-local effect creates new possibilities for quantum information processing.

Mathematically, compound systems are represented by the \emph{tensor product},
$\HH_1\otimes \HH_2$, with
typical element
\[\sum_{i} \lambda_i \cdot \phi_i \otimes \psi_i . \]
\emph{Superposition} encodes \emph{correlation}.

Entanglement is the physical phenomenon underlying
Einstein's `spooky action at a distance'.
Even if the particles are spatially separated, measuring one has an
effect on the state of the other.

Bell's achievement was to turn this puzzling feature of quantum mechanics into a theorem: quantum mechanics is \emph{essentially non-local}.

\subsection{Bell's theorem}

We look again at the empirical model

\begin{center}
\begin{tabular}{l|ccccc}
& $(0, 0)$ & $(1, 0)$ & $(0, 1)$ & $(1, 1)$  &  \\ \hline
$(a, b)$ & $1/2$ & $0$ & $0$ & $1/2$ & \\
$(a, b')$ & $3/8$ & $1/8$ & $1/8$ & $3/8$ & \\
$(a', b)$ & $3/8$ & $1/8$ & $1/8$ & $3/8$ &  \\
$(a', b')$ & $1/8$ & $3/8$ & $3/8$ & $1/8$ & 
\end{tabular}
\end{center}

This can be realized in quantum mechanics,
using a Bell state
\[ \frac{|0 0  \rangle \; + \; | 1 1 \rangle}{\sqrt{2}} , \]
subjected to measurements in the $XY$-plane of the Bloch sphere, at relative angle $\pi/3$.
Systems of this kind have been the subject of extensive experimental investigation, and the predictions of quantum mechanics can be taken to be very highly confirmed.

The question we shall ask, following Bell, is this:
Can we explain these empirical findings by a theory which is \emph{local} and \emph{realistic} in the following sense.
\begin{itemize}
\item A theory is realistic if it ascribes definite values to all observables for every physical state, independently of the activities of any external observers.
\item A theory is local if the outcomes of measurements on spatially separated subsystems depend only on common causal factors. In particular, for space-like separated measurements, the outcomes of the measurements should be independent of each other.
\end{itemize}
We allow for the fact that there may be salient features in the theory determining the outcomes of measurements of which we are not aware. These features are embodied in the notion of \emph{hidden variable}. Thus we take measurement outcomes to be determined, \emph{given some value of this hidden variable}. Moreover, we assume that this hidden variable acts in a local fashion with respect to spatially separated subsystems.

This gives a general notion of theory which behaves in a fashion broadly consistent with classical physical intuitions. The import of Bell's theorem is exactly that \emph{no such theory can account for the empirical predictions of quantum mechanics}. Hence, given that these predictions are so well-confirmed, we must abandon the classical world-view which underpins the assumptions of local realism.

To give a precise statement of Bell's theorem, we must formalize the notion of local hidden variable theory.
We shall give this in a streamlined form, which can be shown to be equivalent to more general definitions which have been considered (see e.g.~Theorem~7.1 in \cite{abramsky2011unified}).

We shall explain this notion in relation to the Bell table given above.
We have a total set of four measurement settings we are considering, two for Alice and two for Bob:
\[ \{ a, a', b, b' \} . \]
A simultaneous assignment of outcomes (0 or 1) to each of these is given by a function
\[ s : \{ a, a', b, b' \} \rTo \{ 0, 1 \} . \]
The fact that an (unknown) hidden variable may be affecting the outcome is captured by saying that we have a probability distribution $d$ on the set of all such functions $s$. Such a probability distribution can be taken to be a canonical form for a hidden variable.

The requirement on this distribution $d$ to be consistent with the empirical data is that, for each of the experimentally accessible combinations of measurement settings
\[ \{ a, b \}, \quad \{ a', b\}, \quad \{ a, b' \}, \quad \{ a', b' \} ,\]
the restriction (or marginalization) of $d$ to this set of measurements yields exactly the observed distribution on outcomes from the corresponding row of the table. For example, we must have $d | \{ a, b \} = d_1$, where
\[ d_1(0,0) = d_1(1,1) = 1/2, \quad d_1(0,1) = d_1(1,0) = 0. \]

A precise statement of a particular instance of Bell's theorem can now be given as follows:
\begin{proposition}
There is no distribution $d$ on the whole set of measurements which yields the observable distributions by restriction.
\end{proposition}
\begin{proof}
Assume for a contradiction that such a distribution $d$ exists.
It will assign a number $X_i \in [0, 1]$ to each $s_i :  \{ a, a', b, b' \} \rTo \{ 0, 1 \}$. There are 16 such functions: we enumerate them by viewing them as binary strings, where the $j$'th bit indicates the assignment of an outcome to the $j$'th measurement, listed as $a, a', b, b'$.

The requirement that this distribution projects onto the distributions in the empirical model translates into 16 equations, one for each entry in the table. It suffices to consider 4 of these equations:
\[
\begin{array}{lclclclcc}
X_1 & + & X_{2} & + & X_{3} & + & X_{4} & = & 1/2 \\
X_2 & + & X_{4} & + & X_{6} & + & X_{8} & = & 1/8 \\
X_3 & + & X_{4} & + & X_{11} & + & X_{12} & = & 1/8 \\
X_1 & + & X_{5} & + & X_{9} & + & X_{13} & = & 1/8 \\
\end{array}
\]
Adding the last three equations yields
\[ X_1 + X_2 + X_3 + 2X_4 + X_5 + X_6 + X_8 + X_9  + X_{11} + X_{12} + X_{13}  \; = \; 3/8 .\]
Since all these terms must be non-negative, the left-hand side of this equation must be greater than or equal to the left-hand side of the first equation, yielding the required contradiction.
\end{proof}

This argument seems very specific to the probabilistic nature of the empirical model.
However, an important theme in the work on no-go theorems in quantum mechanics is to prove results of this kind in a probability-free fashion \cite{greenberger1989going,hardy1992quantum}. This will bring us directly into the arena of relational databases.

\subsection{Hardy's construction}

Hardy's construction \cite{hardy1992quantum} yields a family of empirical models which can be realized in quantum mechanics in similar fashion to the Bell model. However, these families exhibit a stronger form of non-locality property, which does not depend on the probabilities, but only on the support.

We exhibit an example of a support table arising from Hardy's construction.
\begin{center}
\begin{tabular}{l|cccc} 
 &  $(0, 0)$ & $(1, 0)$ & $(0, 1)$ & $(1, 1)$ \\ \hline
$(a, b)$ &  $1$ &  $1$ &  $1$ &  $1$ \\
$(a', b)$ &   $0$ &  $1$ &  $1$ &  $1$ \\
$(a, b')$ &  $0$ &  $1$ &  $1$ & $1$ \\
$(a', b')$ &   $1$ &  $1$ &  $1$ & $0$ \\
\end{tabular}
\end{center}

This arises from a probability table by replacing all positive probabilites by $1$.

Note that we can view this table as encoding a small relational database, as in our discussion in the previous section. There will be four relation tables in this database, one for each of the above rows.
The table corresponding to the first row will have the full set of tuples over $\{ 0, 1 \}$. The tables for the second and third rows will have the form
\begin{center}
\begin{tabular}{| c | c |} \hline
$\alpha$ & $\beta$ \\ \hline \hline
 0 & 1 \\ \hline
 1 & 0 \\ \hline
 1 & 1 \\ \hline
\end{tabular}
\end{center}

while that for the fourth row will have the form
\begin{center}
\begin{tabular}{| c | c |} \hline
$a'$ & $b'$ \\ \hline \hline
 0 & 0 \\ \hline
 0 & 1 \\ \hline
 1 & 0 \\ \hline
\end{tabular}
\end{center}

The property which shows the non-locality of this model is the exact relational analogue of the probabilistic version we considered in relation to the Bell model.

\begin{proposition}
There is no $A$-relation $R$,  where $A = \{a, a', b, b'\}$, which is consistent with the empirical observable supports; that is, for which $R | \{ \alpha, \beta \}$ yields the relational table for all $\{ \alpha, \beta \}$, $\alpha \in \{ a, a'\}$, $\beta \in \{ b, b' \}$.
\end{proposition}

In database language, this says exactly that there is no `universal relation' on the whole set of attributes which yields each of the `observable relations' by projection.

\begin{proof}
We argue similarly to the case of the Bell model, except that we are now working over the Boolean semiring rather than the non-negative reals. The existence of a relation $R$ thus reduces to a Boolean satisfiability problem. An equation $\sum_i X_i = 1$ simply asserts the disjunction of the boolean variables, while an equation $\sum_i X_i = 0$ asserts the conjunction of the negated variables. Again it suffices to consider four of the equations which can be read off the Hardy table:
\[ \begin{array}{ccccccc}
X_1 & \OR & X_2 & \OR & X_3 & \OR & X_4 \\
 \neg X_1 & \AND & \neg X_3 & \AND & \neg X_ 5 & \AND & \neg X_7 \\
  \neg X_1 & \AND & \neg X_2 & \AND & \neg X_ {9} & \AND & \neg X_{10} \\
   \neg X_4 & \AND & \neg X_8 & \AND & \neg X_ {12} & \AND & \neg X_{16} 
\end{array}
\]
Since every disjunct in the first formula appears as a negated conjunct in one of the other three formulas, there is no satisfying assignment.
\end{proof}

There is a precise sense in which the Hardy result is stronger than the Bell result.
In fact, we have the following.

\begin{proposition}
If an empirical model has a local hidden-variable model in the probabilistic sense, then its support table has a universal relation. Thus failure to have a universal relation implies failure to have a local hidden-variable model in the probabilistic sense.
\end{proposition}
\begin{proof}
This follows simply from the fact that the map  from the non-negative reals to the booleans which takes all non-zero elements to 1 is a semiring homomorphism.
\end{proof}

The converse to this result is false. For example, the support table arising from the Bell model \emph{does} have a universal relation, as can easily be verified.

Note that the Hardy table, and indeed all such tables arising from quantum mechanics, satisfies the compatibility condition which we discussed in Section~\ref{sheafsec}.
In fact, compatibility corresponds precisely to the physical condition of \emph{no-signalling}, and the fact that quantum models satisfy the condition is exactly the content of the No-Signalling theorem of quantum mechanics. See \cite{abramsky2011unified} Section~8 for an extended discussion of this point.

\section{No-Go theorems, global sections and universal relations}

We shall now develop a more general perspective on the results we have discussed in the previous section.

Following the geometric language we introduced in Section~\ref{sheafsec}, we see that the existence of a hidden-variable model is equivalently expressed as the existence of a \emph{global section} which glues together the family of empirical accessible distributions or relations.

Thus non-locality and related no-go results can be understood in terms of \emph{obstructions to the existence of global sections}, a central issue in the pervasive applications of sheaves in geometry, topology, analysis and number theory.

In terms of databases, such results can be understood as expressing obstructions to the existence of universal relations for given instances of the database.

We shall now discuss two further types of no-go results, which can be understood in terms of yet stronger forms of obstruction.

\subsection{Strong Contextuality}

If we consider the argument for the Hardy construction, it can be understood as saying that there is no relation over the global tuples which `covers' all (and only) the observable tuples.
But now suppose we consider a much weaker requirement: we simply ask for \emph{one global tuple} which projects consistently into all the relations in the database instance.

Note that the Hardy model does meet this condition. The global assignment
\[ \{ a \mapsto 1, \; a' \mapsto 0, \; b \mapsto 1, \; b' \mapsto 0 \} \]
does project consistently into the support table for this model. The Bell model similarly meets this condition.

If even this, much weaker requirement \emph{cannot} be met, then we have a much stronger form of no-go theorem. We say that such a situation exhibits \emph{strong contextuality}.

The question now arises: are there models coming from quantum mechanics which are strongly contextual in this sense?

We shall now show that the well-known GHZ models \cite{greenberger1989going}, of type $(n, 2, 2)$ for all $n > 2$, are strongly contextual. This will establish a strict hierarchy
\[ \mbox{Bell} < \mbox{Hardy} < \mbox{GHZ} \]
of increasing strengths of obstructions to non-contextual behaviour for these salient models.

The GHZ model of type $(n, 2, 2)$ can be specified as follows.
We label the two measurements at each part as $\mxi$ and $\myi$, and the outcomes as $0$ and $1$.
For each context $C$, every $s$ in the support of the model satisfies the following conditions:
\begin{itemize}
\item If the number of $Y$ measurements in $C$ is a multiple of 4, the number of $1$'s in the outcomes specified by $s$ is even.
\item If the number of $Y$ measurements  is $4k+2$, the number of $1$'s in the outcomes is odd.\end{itemize}
A model with these properties can be realized in quantum mechanics, using the GHZ state
\[ \frac{\ket{0 \cdots 0} \; + \; \ket{1 \cdots 1}}{\sqrt{2}} . \]

\begin{proposition}
The GHZ models are strongly contextual, for all $n \geq 3$.
\end{proposition}
\begin{proof}
We consider the case where $n = 4k$, $k \geq 1$.
Assume for a contradiction that we have a global section $s \in S_e$ for the GHZ model $e$.

If we take $Y$ measurements at every part, the number of $1$ outcomes under the assignment is even.
 Replacing any two $Y$'s by $X$'s changes the residue class mod $4$ of the number of $Y$'s, and hence must result in the opposite parity for the number of $1$ outcomes under the assignment.
Thus for any $\myi$, $\myj$ assigned the \textit{same} value, if we substitute $X$'s in those positions they must receive \textit{different} values under $s$. Similarly, for any  $\myi$, $\myj$ assigned different values, the corresponding  $\mxi$, $\mxj$ must receive the same value.

Suppose firstly that not all $\myi$ are assigned the same value by $s$.
Then for some $i$, $j$, $k$, $\myi$ is assigned the same value as $\myj$, and $\myj$ is assigned a different value to $\myk$. Thus $\myi$ is also assigned a different value to $\myk$. Then $\mxi$ is assigned the same value as $\mxk$, and $\mxj$ is assigned the same value as $\mxk$. By transitivity, $\mxi$ is assigned the same value as $\mxj$, yielding a contradiction.

The remaining cases are where all $Y$'s receive the same value. Then any pair of $X$'s must receive different values. But taking any 3 $X$'s, this yields a contradiction, since there are only two values, so some pair must receive the same value.

The case when $n = 4k+2$, $k \geq 1$,  is proved in the same fashion, interchanging the parities.
When $n \geq 5$ is odd, we start with a context containing one $X$, and again proceed similarly.

The most familiar case, for $n=3$, does not admit this argument, which relies on having at least 4 $Y$'s in the initial configuration.  However, for this case one can easily adapt the well-known argument of Mermin using `instruction sets' \cite{mermin1990quantum} to prove strong contextuality. This uses a case analysis to show that there are 8 possible global sections satisfying the parity constraint on the 3 measurement combinations with 2 $Y$'s and 1 $X$; and all of these violate the constraint for the $XXX$ measurement.
\end{proof}


\subsection{The Kochen-Specker theorem}

Kochen-Specker-type theorems \cite{kochen1975problem} can be understood as \emph{generic strong contextuality results}. In database terms, they say that, if the database schema has a certain combinatorial structure, then \emph{every instance} satisfying some conditions is strongly contextual.
This can be interpreted in the quantum context in such a way that the conditions will be satisfied by \emph{every} quantum state, and hence we obtain a state-independent form of strong contextuality result.

The condition which is typically imposed on the instances, assuming that the possible data values for each attribute lie in $\{ 0, 1 \}$, is that \emph{every tuple contains exactly one $1$}. If we think in terms of satisfiability, this corresponds to a `POSITIVE ONE-IN-$k$-SAT' condition.

To show that the Kochen-Specker result holds is exactly to show that there is no satisfying assignment for the corresponding set of clauses.

The simplest example of this situation is the `triangle', \ie the schema with  elements
\[  \{a, b\}, \{ b, c \}, \{ a, c \}  . \]
However, this example cannot be realized in quantum mechanics \cite{abramsky2011unified}.

An example which can be realized in quantum mechanics, where $\Att$ has 18 elements, and there are 9 sets in the database schema,  each with four elements, such that each element of $\Att$ is in two of these, appears in the 18-vector proof of the Kochen-Specker Theorem in \cite{cabello1996bell}.
\begin{center}
\begin{tabular}{|c|c|c|c|c|c|c|c|c|} \hline
$U_1$ & $U_2$ &  $U_3$ & $U_4$ & $U_5$ & $U_6$ & $U_7$ & $U_8$ & $U_9$ \\ \hline\hline
$A$ & $A$ & $H$ & $H$ & $B$ & $I$ & $P$ & $P$  & $Q$ \\ \hline
$B$ & $E$ & $I$ & $K$ & $E$ & $K$ & $Q$ & $R$ & $R$  \\ \hline
$C$ & $F$ & $C$ & $G$ & $M$ &  $N$ & $D$ & $F$ & $M$  \\ \hline
$D$ & $G$ & $J$ & $L$ & $N$  & $O$ & $J$ & $L$ & $O$  \\ \hline
\end{tabular}
\end{center}

Here the schema is $\Sigma = \{ U_1, \ldots , U_9 \}$.

We shall  give a simple combinatorial condition on the schema $\Sigma$ which is implied by the existence of a global section $s$ satisfying  the `POSITIVE ONE-IN-$k$-SAT' condition. Violation of this condition therefore suffices to prove that no such global section exists.

For each $a \in \Att$, we define
\[ \Sigma(a) := \{ A \in \Sigma \mid a \in A \} . \]
\begin{proposition}
If a global section satisfying the condition exists, then every common divisor of $\{ \card{\Sigma(a)} \mid a \in \Att \}$ must divide $\card{\Sigma}$.
\end{proposition}
\begin{proof}
Suppose there is a global section $s : \Att \rarr \{ 0, 1 \}$ satisfying the condition. Consider the set $X\subseteq \Att$ of those $a$ such that $s(a) = 1$. Exactly one element of $X$ must occur in every $A \in \Sigma$. Hence there is a partition of $\Sigma$ into the subsets $\Sigma(a)$ indexed by the elements of $X$. Thus
\[ \card{\Sigma} = \sum_{a \in X} \card{\Sigma(a)} . \]
It follows that, if there is a common divisor of the numbers $\card{\Sigma(a)}$,  it must divide $\card{\Sigma}$.
\end{proof}

For example, if every $a \in \Att$ appears in an even number of  elements of $\Sigma$, while $\Sigma$ has an odd number of elements, then there is no global section.
This corresponds to the `parity proofs' which are often used in verifying Kochen-Specker-type results \cite{cabello1996bell,WaegellAravind2011}. For example, in the 18-attribute schema with 9 relations given above, each attribute appears in two relations in the schema; hence the argument applies.

For further discussion of these ideas, including connections with graph theory, see \cite{abramsky2011unified}.

\section{Further Directions}

We mention some further directions for developing the connections between databases and the study of non-locality and contextuality in quantum mechanics.

\begin{itemize}
\item We may consider conditions on the database schema which guarantees that global sections can be found. The important notion of \emph{acyclicity} in database theory \cite{beeri1983desirability} is relevant here.
On the probabilistic side there is a result by Vorob'ev \cite{vorob1962consistent} (motivated by game theory), which gives necessary and sufficient combinatorial conditions on a schema for \emph{any} assignment of probability distributions on the tuples for each relation in the schema to have a global section; that is, for a universal relation in the probabilistic sense to always exist for any probabilistic instance of the database.
Rui Soares Barbosa (personal communication) has shown that the Vorob'ev condition is equivalent to acyclicity in the database sense.
This provides another striking connection between database theory and the theory of quantum non-locality and contextuality.

\item A logical approach to Bell inequalities in terms of logical consistency conditions is developed in \cite{abramsky2012logical}. It would be interesting to interpret and apply this notion of Bell inequalities in the database context.

\item The tools of sheaf cohomology are used to characterize the obstructions to global sections in a large family of cases in \cite{abramsky2011cohomology}. In principle, these sophisticated tools can be applied to databases. There may be interesting connections with acyclicity in the database sense.
\end{itemize}

We can summarise the connections which we have exposed between database theory and quantum non-locality and contextually in the following table:

\begin{center}
\begin{tabular}{l|l} 
Relational databases & measurement scenarios \\ \hline
attribute & measurement \\
set of attributes defining a relation table & compatible set of measurements \\
database schema & measurement cover \\
tuple & local section (joint outcome) \\
relation/set of tuples & boolean distribution on joint outcomes \\
universal relation instance & global section/hidden variable model \\
acyclicity & Vorob'ev condition
\end{tabular}
\end{center}

\paragraph{Acknowledgements}

Discussions with and detailed comments by Phokion Kolaitis are gratefully acknowledged.
Leonid Libkin also gave valuable feedback.
This paper was written while in attendance at the program on `Semantics and Syntax: the legacy of Alan Turing' at the Isaac Newton Institute, Cambridge, April--May 2012.

\bibliographystyle{plain}

\bibliography{bdbib}
\end{document}